\documentclass[prl,twocolumn,superscriptaddress]{revtex4-1}

\usepackage{graphics}
\usepackage{graphicx}
\usepackage{amsthm} 
\usepackage{amsbsy} 
\usepackage{amsmath,amsfonts,amssymb,times,color}
\usepackage{enumerate}
\usepackage{hyperref}
\usepackage{natbib}
\usepackage{braket}

\newtheorem{thm}{Theorem}
\newtheorem{cor}[thm]{Corollary}

\theoremstyle{remark}

\theoremstyle{definition}
\newtheorem{dfn}[thm]{Definition}

\newcommand{\be}{\begin{equation}}
\newcommand{\ee}{\end{equation}}
\newcommand{\ba}{\begin{eqnarray}}
\newcommand{\ea}{\end{eqnarray}}
\newcommand{\ban}{\begin{eqnarray*}}
\newcommand{\ean}{\end{eqnarray*}}
\newcommand{\bit}{\begin{itemize}}
\newcommand{\eit}{\end{itemize}}
\newcommand{\ben}{\begin{enumerate}}
\newcommand{\een}{\end{enumerate}}

\newcommand\abs[1]{\left|#1\right|}
\newcommand{\born}[2]{\abs{\braket{#1|#2}}^2}

\def\ketbra[#1]#2{\mathinner{\vert{#1}\rangle\langle{#2}}\vert}

\newcommand{\de}[1]{\,\mathrm{d}#1}

\definecolor{red}{rgb}{0.9,0,0}
\definecolor{green}{rgb}{0,0.8,0}
\definecolor{blue}{rgb}{0,0,0.8}
\definecolor{cautionred}{rgb}{1.0,0,0}

\definecolor{maroon}{rgb}{0.7,0,0}

\definecolor{ngreen}{rgb}{0.3,0.7,0.3}

\definecolor{golden}{rgb}{0.8,0.6,0.1}

\begin{document}
\title{No $\psi$-epistemic model can fully explain the indistinguishability of quantum states}
\author{Jonathan Barrett}
\affiliation{Department of Computer Science, University of Oxford, UK}
\author{Eric G. Cavalcanti}
\affiliation{School of Physics, The University of Sydney, Australia}
\affiliation{Department of Computer Science, University of Oxford, UK}
\author{Raymond Lal}
\affiliation{Department of Computer Science, University of Oxford, UK}
\author{Owen J. E. Maroney}
\affiliation{Faculty of Philosophy, University of Oxford, UK}
\date{\today}

\begin{abstract}

According to a recent no-go theorem (M.~Pusey, J.~Barrett and T.~Rudolph, Nature Physics {\bf 8} 475 (2012)), models in which quantum states correspond to probability distributions over the values of some underlying physical variables must have the following feature: the distributions corresponding to distinct quantum states do not overlap. This is significant because if the distributions do not overlap, then the quantum state itself is encoded by the physical variables. In such a model, it cannot coherently be maintained that the quantum state merely encodes information about underlying physical variables. The theorem, however, considers only models in which the physical variables corresponding to independently prepared systems are independent. This work considers models that are defined for a single quantum system of dimension $d$, such that the independence condition does not arise. We prove a result in a similar spirit to the original no-go theorem, in the form of an upper bound on the extent to which the probability distributions can overlap, consistently with reproducing quantum predictions. In particular, models in which the quantum overlap between pure states is equal to the classical overlap between the corresponding probability distributions cannot reproduce the quantum predictions in any dimension $d \geq 3$. The result is noise tolerant, and an experiment is motivated to distinguish the class of models ruled out from quantum theory.
\end{abstract}
\maketitle


No-go theorems such as Bell's \cite{Bell1964} are of central importance to our understanding of quantum mechanics. Bell's theorem shows that locally causal models must make different predictions from quantum theory. In addition to the fundamental significance of this result, Bell's theorem has applications in quantum information processing, most notably in device-independent cryptography and randomness generation \cite{Ekert1991, Barrett2005, Cavalcanti2012, Pironio2010}. 

Recently, a number of new no-go results have been derived, addressing a different question than whether nature can be described by a locally causal theory.  The question concerns whether the quantum state should be viewed as a description of the physical state of a system, or as an observer's information about the system. Many authors (see, e.g., Refs.~\cite{Spekkens2005, Bartlett2011, Fuchs2013}, and references therein) have argued for the latter, pointing out, for example, that quantum collapse is analogous to Bayesian updating of a classical probability distribution when new data is obtained, or that the indistinguishability of non-orthogonal quantum states is analogous to the indistinguishability of overlapping probability distributions. Ref.~\cite{Pusey2012}, following Ref.~\cite{Harrigan2010}, considers models of a specific form, in which the quantum state corresponds to a probability distribution over some set of underlying physical states, hence can be thought of as representing an observer's partial information about the physical state. It is shown that such models cannot recover the quantum predictions unless the distributions are disjoint for distinct quantum states. Roughly speaking, if the assumptions of Ref.~\cite{Pusey2012} are accepted, then the quantum state must describe some part of reality.

One assumption of Ref.~\cite{Pusey2012} is that the physical states are uncorrelated for independently prepared systems. It is interesting to investigate what can be established without this assumption. Various works have investigated what can be concluded by considering measurements on a single system only, i.e, without any assumption about independent systems \cite{Lewis2012,Maroney2012a,Leifer2013, Patra2013, Aaronson2013}. Here, we consider a single quantum system, and derive bounds on the extent to which the probability distributions corresponding to distinct quantum states can overlap. We show that what we call \emph{maximally $\psi$-epistemic models}, in which the overlap of the probability distributions is large enough to explain fully the indistinguishability of quantum states, must make different predictions from quantum theory for Hilbert space dimension $d\geq 3$. Our result is noise-tolerant, allowing for experimental tests to rule out this class of models. Furthermore, we show that as $d\rightarrow\infty$, any model recovering quantum predictions must become \emph{arbitrarily bad} at explaining quantum state indistinguishability.

\textit{Non-orthogonality and epistemic states.}
Non-orthogonal quantum states cannot be distinguished with certainty in a single shot. This is sometimes regarded as a distinctly quantum phenomenon, but of course a similar thing is true of classical probability distributions. Consider a standard deck of 52 playing cards and a shuffling/drawing machine with two settings: with the first setting, a red card is drawn at random, and with the second setting, the card is a randomly chosen ace. The two settings correspond to probability distributions $p$ and $q$ such that $p=\frac{1}{26}$ for all red cards and $q=\frac{1}{4}$ for each ace. Given a single card drawn from the pack, and asked to determine under what setting the machine was operating, one cannot succeed with certainty. The reason is simply that the distributions $p$ and $q$ overlap, e.g., $p$ and $q$ are both nonzero for the ace of hearts.

This suggests that the inability to distinguish non-orthogonal quantum states could be explained analogously. In that case, two quantum states would be indistinguishable in a single-shot experiment because they would correspond to overlapping distributions over states of reality. The aim of this work is to explore the extent to which such an explanation is even possible, consistently with the quantum predictions.

\textit{Ontological models for quantum theory.} To formalize this idea, we shall use the framework  of \em ontological models \em \cite{Spekkens2005, Harrigan2010}. This framework assumes that when a physical system has been prepared in the quantum state $\ket{\psi}$, it is actually in an \em ontic state \em $\lambda$, which we can think of as the `state of reality'. An ontological model assigns to each quantum state $\ket{\psi}$ an \em epistemic state \em  $\mu_\psi$, which is a probability distribution over the set of ontic states $\Lambda$, and represents our ignorance about which ontic state $\lambda$ the system is in. Since an epistemic state is a probability distribution, it must satisfy
\begin{eqnarray} \label{eq:mu}
\mu_\psi(\lambda) \geq 0 & \quad\textrm{   and   }\quad & \int \mu_\psi(\lambda) \de{\lambda} =1.
\end{eqnarray}
The framework assumes that when a measurement is performed, the probability for a given outcome depends only on the ontic state $\lambda$. Hence for a measurement $M$, and outcome $f$, an ontological model assigns a \em response function\em, which yields the probability $\xi_M(f | \lambda)$ of obtaining the outcome $f$ in the state $\lambda$, and we have:
\begin{eqnarray}\label{eq:xi}
\xi_M(f | \lambda) \geq 0 & \quad\textrm{   and   }\quad & \sum_f \xi_M (f | \lambda) =1.
\end{eqnarray}
To reproduce the predictions of quantum theory, response functions must satisfy
\be\label{eq:bornrule}
\int_\Lambda \xi_M(f|\lambda)\mu_\psi(\lambda) \,\mathrm{d}\lambda=\abs{\braket{f|\psi}}^2
\ee
for all $\ket{\psi}$ and $f$.

Standard distance measures, defined on probability distributions and quantum states, will be useful in the following. For distributions $p(x)$ and $q(x)$, the \emph{classical trace distance} is
\be
\delta_C(p,q):=\frac{1}{2}\int \abs{p(x)-q(x)}\,\mathrm{d}x \,.
\ee
This quantity has an operational interpretation. Suppose that the distributions $p(x)$ and $q(x)$ are associated with two different preparations of the variable $x$ (as with the cards above), and suppose that equal a priori probabilities are assigned to the two preparations.  The probability of correctly guessing the preparation, given a single sample of $x$ is $1/2(1+\delta_C(p,q))$. 

In the quantum case, the \emph{quantum trace distance}, for pure states, is given by
\be\label{eq:quantum_trace}
\delta_Q(\psi,\phi)=\sqrt{1-\abs{\braket{\psi|\phi}}^2} \,.
\ee
If one of a pair of quantum states $|\psi\rangle$ or $|\phi\rangle$ is prepared with equal probability, then, by using an optimal measurement, the probability of correctly identifying which state has been prepared is $1/2(1+\delta_Q(\psi,\phi))$. 

Define the {\it classical overlap} of two distributions $p$ and $q$ as
\be
\omega_C(p,q) := 1 - \delta_C(p,q) \, = \int \min \{p(x),q(x)\} \,\mathrm{d}x.
\ee
Similarly, for quantum states $|\psi\rangle$ and $|\phi\rangle$, let the \emph{quantum overlap} be given by
\be
\omega_Q(\psi,\phi) := 1 - \delta_Q(\psi,\phi).
\ee

Following Ref.~\cite{Harrigan2010}, 
\begin{dfn}
An ontological model is  \emph{$\psi$-epistemic} if there exists at least one pair of distinct quantum states, $\ket{\psi}$ and $\ket{\phi}$, such that the corresponding epistemic states $\mu_\psi$ and $\mu_\phi$ have nonzero overlap, i.e., $\omega_C(\mu_\psi,\mu_\phi) > 0$. If a model is not $\psi$-epistemic, then it is \emph{$\psi$-ontic}. \footnote{The terminology derives from the idea that if $\mu_\psi$ and $\mu_\phi$ do not overlap for any pair of distinct quantum states, then distinct quantum states always refer to distinct states of reality. In this case, the quantum state itself can be regarded as a part of reality (hence it is \emph{ontic}). In the context of ontological models, it must be the case that $\mu_\psi$ and $\mu_\phi$ typically overlap if it is to be maintained that the quantum state represents only information about reality (hence is \emph{epistemic}). (As it stands, the definition of $\psi$-epistemic is very weak because it is only required that the distributions overlap for one pair of quantum states. This serves to make results that rule out $\psi$-epistemic models stronger.)}
\end{dfn}

Hardy \cite{hardyprivcomm} raised the question of whether $\psi$-epistemic models could reproduce the predictions of quantum theory. Ref.~\cite{Pusey2012} then showed that under an assumption to do with the independence of separately prepared systems, they cannot. The assumption is that when two quantum systems are prepared independently, they can be assigned separate ontic states $\lambda_1$ and $\lambda_2$, and that the joint distribution satisfies $\mu_{\psi\otimes\phi}(\lambda_1,\lambda_2) = \mu_{\psi}(\lambda_1)\times \mu_{\phi}(\lambda_2)$. Various works since have explored the possibilities for ontological models for single systems, i.e., without this assumption. Ref.~\cite{Lewis2012} shows that $\psi$-epistemic models exist for quantum systems of arbitrary dimension. Ref.~\cite{Aaronson2013} goes further, demonstrating that for a quantum system of arbitrary dimension, a $\psi$-epistemic model exists with the additional property that $\omega_C(\mu_\psi,\mu_\phi) > 0$ for every pair of non-orthogonal states $|\psi\rangle$ and $|\phi\rangle$. Refs.~\cite{Aaronson2013, Hardy2012, Patra2013} show that $\psi$-epistemic models do not exist, given various additional assumptions. In Refs.~\cite{Maroney2012a,Leifer2013}, the question is raised whether $\psi$-epistemic models can reproduce quantum predictions given an assumption about the extent to which the epistemic states overlap.

Refs.~\cite{Maroney2012a,Leifer2013} are the most direct precursors to this work, since here we are also concerned with the extent to which the distributions $\mu_\psi$ and $\mu_\phi$ can overlap in models which recover the predictions of quantum theory. An advantage of the present work is that we use distance measures that are robust under small variations, hence our results are noise tolerant and subject to experimental test.

The following is an easy theorem, previously noted in Ref.~\cite{pbriontrap}.
\begin{thm} \label{thm:delta_ineq}
In any ontological model that recovers the predictions of quantum theory, 
\begin{equation} \label{eq:delta_ineq}
\omega_C(\mu_\psi,\mu_\phi) \leq \omega_Q(\psi, \phi)  \quad \forall \psi,\phi \, .
\end{equation}
\end{thm}
\textit{Proof.} Consider the optimal measurement for distinguishing two quantum states. Success occurs with probability $P_Q:=1-\omega_Q(\psi,\phi)/2$. Given the ontic state $\lambda$, the maximum probability to correctly guess which preparation was performed is given by $P_C:=1-\omega_C(\mu_\psi,\mu_\phi)/2$. But in an ontological model the output of the quantum measuring device depends only on the ontic state $\lambda$, thus $P_Q\leq P_C$ since $P_Q$ cannot be larger than what one would get by optimally using the information encoded in $\lambda$. $\Box$

\begin{dfn}
An ontological model is \emph{maximally $\psi$-epistemic} if and only if  for all pairs of states, $\omega_C(\mu_\psi,\mu_\phi) = \omega_Q(\psi,\phi)$. \footnote{The term \emph{maximally $\psi$-epistemic} was also used in Refs.~\cite{Maroney2012a} and \cite{Leifer2013}. In the latter, it has a slightly different definition.} 
\end{dfn}
The motivation for this terminology is that, as we have already argued, the impossibility of discriminating non-orthogonal quantum states would be explained in a natural way if the two quantum states sometimes correspond to the same state of reality. But this explanation would not be satisfying if the quantum and classical overlaps were not equal. For then, the two classical distributions could in principle be better discriminated by a device with access to $\lambda$, and some additional explanation must be adduced as to why the two quantum states are hard to distinguish. In a maximally $\psi$-epistemic model, on the other hand, the difficulty of discriminating non-orthogonal quantum states is completely and quantitatively explained by the difficulty of discriminating the corresponding epistemic states.

\textit{Ruling out maximally $\psi$-epistemic models.} Our results rule out maximally $\psi$-epistemic models for quantum systems of dimension $d \geq 3$, and are noise-tolerant. For $d=2$, an ontological model due to Kochen and Specker \cite{Kochen1967} can be shown to be maximally $\psi$-epistemic \cite{peterlewisprivcomm}. 

The case of three-dimensional systems, and an analysis designed to account for experimental noise, will follow below. First, consider systems of dimension $d\geq 4$. 
\begin{thm}\label{maintheorem}
Suppose that an ontological model reproduces the quantum predictions for a system of dimension $d\geq 4$, and that 
\[
\omega_C(\mu_\psi,\mu_\phi) \geq k \,\omega_Q(\psi,\phi) \quad \forall \psi,\phi ,
\]
for some constant $k$. Then $k < 4 / (d-1)$. If $d$ is power prime, then $k < 2/d$.
\end{thm}

\begin{proof}  
Using terminology introduced by Caves, Fuchs and Schack \cite{Caves2002}, three pure states $\ket{a}$, $\ket{b}$ and $\ket{c}$, are \emph{PP-incompatible} if there exists an orthonormal basis $\{\ket{f_i}\}_{i=1}^3$ for the subspace spanned by $\ket{a}$, $\ket{b}$ and $\ket{c}$ such that $\braket{f_1|a} = 0$, $\braket{f_2|b} = 0$, and $\braket{f_3|c} = 0$. Ref.~\cite{Caves2002} shows the following. Let ${x_1:=\abs{\braket{a|b}}^2}$, ${x_2:=\abs{\braket{b|c}}^2}$ and ${x_3:=\abs{\braket{c|a}}^2}$. Then $\ket{a}$, $\ket{b}$ and $\ket{c}$ are PP-incompatible if and only if \footnote{There is a typographical error in Ref.~\cite{Caves2002}: the second inequality was there written as a strict inequality, but the non-strict inequality is correct.}
\ba\label{eq:PP_incompatible}
x_1 + x_2 + x_3 & < & 1\, \nonumber \\
(x_1 + x_2 + x_3 - 1)^2 & \geq & 4 \, x_1 x_2 x_3 \,.
\ea

Recall that a pair of bases $\{\ket{a_i}\}_i$ and $\{\ket{b_j}\}_j$ is \em mutually unbiased \em if $\abs{\braket{a_i | b_j}}^2 = 1/d$ for all $i,j$, where $d$ is the Hilbert-space dimension. If $d$ is power prime, then there exist $d+1$ mutually unbiased bases \cite{Klappenecker2004}. Let $\ket{c}$ be an element of one such basis, and for $i,\gamma\in\{1,\dots,d\}$, let the $d$ remaining bases be $\{\ket{e_i^\gamma}\}_i$, where $\gamma$ ranges over the distinct bases, and $i$ over the elements within a basis. For $\alpha\ne \beta$ and $d\geq 4$, the set $\{\ket{e^\alpha_i}, \ket{e^\beta_j}, \ket{c}\}$ is PP-incompatible by Eq.~\eqref{eq:PP_incompatible}.

Now, consider an ontological model for systems of dimension $d\geq 4$ with $d$ power prime. From the PP-incompatibility of $\{\ket{e^\alpha_i}, \ket{e^\beta_j}, \ket{c}\}$, it follows that there exists a measurement $M$ with outcomes $f_i$, $i=1,\ldots,4$  such that
\be\label{eq:PP_model1}
\int_\Lambda \xi_M(f_1 |\lambda)\mu_{e^\alpha_i}(\lambda) \,\mathrm{d}\lambda=\abs{\braket{f_1|e^\alpha_i }}^2 =  0,
\ee
\be\label{eq:PP_model2}
\int_\Lambda \xi_M(f_2 |\lambda)\mu_{e^\beta_j}(\lambda) \,\mathrm{d}\lambda=
\int_\Lambda \xi_M(f_3 |\lambda)\mu_{c}(\lambda) \,\mathrm{d}\lambda=  0.
\ee
and the outcome $f_4$ is a projector onto the orthogonal subspace and has zero probability on each of the three states.

Assume for contradiction that there is a subset $\Lambda^*\subseteq \Lambda$ of non-zero measure such that $\mu_{e^\alpha_i}(\lambda), \mu_{e^\beta_j}(\lambda), \mu_{c}(\lambda) > 0$ for all $\lambda\in\Lambda^*$. Eq.~\eqref{eq:PP_model1} and Eq.~\eqref{eq:PP_model2} then imply that for some $\lambda$, $\xi_M(f_1|\lambda)=\xi_M(f_2|\lambda)=\xi_M(f_3|\lambda)=0$. But this, along with the fact that $f_4$ has probability zero on all three states, contradicts Eq.~\eqref{eq:xi}. For quantum state $|\psi\rangle$, let $\Lambda_\psi$ denote the support of the distribution $\mu_\psi$. It follows that for any $\alpha \neq \beta$, and for any $i,j$, $\Lambda_{e^\alpha_i} \cap \Lambda_{e^\beta_j} \cap \Lambda_c$ is a set of measure zero.

Now, for any pair of distributions $\mu_\psi$ and $\mu_\phi$,
\be\label{eq:min_bound}
\int_{\Lambda_\phi}\mu_\psi(\lambda)\,\mathrm{d}\lambda \geq \omega_C(\mu_\psi, \mu_\phi) \,.
\ee
Assume that the ontological model satisfies $\omega_C(\mu_\psi,\mu_\phi) \geq k \,\omega_Q(\psi,\phi)$ for all pairs of states. Then for any $\gamma, i$,
\be\label{eq:inequality} 
\int_{\Lambda_{e^\gamma_i}}\mu_c(\lambda)\,\mathrm{d}\lambda \ge k\, \left(1-\sqrt{1-1/d}\right) \, .
\ee

For $i \neq j$ the vectors $\ket{e^\gamma_i}$ and $\ket{e^\gamma_j}$ are orthogonal, and can be distinguished by a single shot measurement. It follows that $\Lambda_{e^\gamma_i} \cap \Lambda_{e^\gamma_j}$ is a set of measure zero. Hence
\be
\int_{\bigcup\limits_i\Lambda_{e^\gamma_i}}\mu_c(\lambda)\,\mathrm{d}\lambda \ge  d\, k\, \left(1-\sqrt{1-1/d}\right)  \,.
\ee
Using the fact that $\Lambda_{e^\alpha_i} \cap \Lambda_{e^\beta_j} \cap \Lambda_c$ is a set of measure zero, 
\be 
\int_{
\bigcup\limits_\gamma \bigcup\limits_i \Lambda_{e^\gamma_i}
}
\mu_c(\lambda)\,\mathrm{d}\lambda  \ge  d^2\, k\, \left(1-\sqrt{1-1/d}\right)   \,.
\ee
This gives
\be
\label{eq:final_4d}
k \le \frac{1}{d}\left(1+\sqrt{1-1/d}\right) < \frac{2}{d} \,. 
\ee

The result for a system of arbitrary dimension $d \geq 4$ now follows immediately. Consider a $d'$-dimensional subspace, where $d' \leq d$ and $d'$ is power prime. The theorem applies to ontological models that recover the quantum predictions for preparations and measurements within this subspace. Hence any ontological model for the $d$-dimensional system must have $k < 2/d'$. Bertrand's Postulate states that for every natural number $n \geq 2$, there is a prime between $n$ and $2n$ \cite{Ramanujan1919}. Choosing $n = \lfloor d/2 \rfloor $ yields $k < 4 / (d-1)$.
\end{proof}
\begin{cor}\label{maincorollary}
No maximally epistemic ontological model can reproduce the quantum predictions for a system of dimension $d\geq 4$.
\end{cor}
\begin{proof} For a maximally epistemic ontological model, the antecedent of Theorem \ref{maintheorem} holds with $k=1$. But then we conclude that $k<1$, reaching a contradiction.
\end{proof}
Moreover, the upper bound on $k$ asymptotically tends to zero, meaning that as $d \rightarrow \infty$, every ontological model will assign a ratio between the classical and quantum overlaps tending to zero for at least some pairs of quantum states.

\textit{Note.} Examination of the proof shows that it is possible to state a stronger result (which we have left out of Theorem~\ref{maintheorem} for simplicity). Suppose that $k(\psi,\phi)$ is defined so that for each pair of states, $\omega_C(\mu_\psi,\mu_\phi) = k(\psi,\phi) \, \omega_Q(\psi,\phi)$. Then, a bound can be derived on the value of $k(\psi,\phi)$, averaged over the states used in the proof: 
\be
\sum_{\alpha, i} \frac{k(c,e^\alpha_i)}{d^2} < \frac{4}{d-1} \, ,
\ee
where $|c\rangle$, $|e^\alpha_i\rangle$ all lie within a power prime-dimensional subspace. Since $|c\rangle$ can be chosen to be an arbitrary state (by applying the same unitary to $|c\rangle$ and all the other states in the proof, thus maintaining their overlaps), this implies that for every quantum state in $d\geq4$, there exists a finite set of states such that the average ratio of classical and quantum overlaps is bounded as above. 

\textit{The noisy case.} In a real experiment, observed relative frequencies will not exactly match the quantum predictions, hence if the experiment is to rule out a class of ontological models, it is necessary to consider models that only approximately reproduce quantum predictions. Suppose that an experiment is carried out in which quantum systems are repeatedly prepared and then measured. Each time, the preparation is (intended to be) of a pure state chosen at random from the set of mutually unbiased bases employed in the proof of Theorem~\ref{maintheorem}. The measurement is (intended to be) either a projective measurement onto one of these bases, or a projective measurement $M$ with outcomes $f_1,\ldots, f_4$, chosen so that $\langle f_1 | e^\alpha_i\rangle = \langle f_2 | e^\beta_j\rangle = \langle f_3 | c \rangle = 0$ for some triple $(|e^\alpha_i\rangle , | e^\beta_j \rangle , |c\rangle )$, with $f_4$ corresponding to a projector onto the orthogonal subspace.

Let $R[g|\psi]$ be the relative frequency with which outcome $g$ is observed when the preparation is $\psi$. Quantum theory predicts, for example, that if $\langle f_1|e^\alpha_i\rangle=0$, and the experiment is carried out perfectly, then $R[f_1 | e^\alpha_i]$ will be zero, while noise will ensure that $R[f_1 | e^\alpha_i]$ is typically greater than zero. The following analysis is designed to take this noise into account. For simplicity, we assume that the measurement is perfectly aligned in the three-dimensional subspace spanned by $(|e^\alpha_i\rangle , | e^\beta_j \rangle , |c\rangle )$, ignoring the possibility that the outcome $f_4$ occurs. We also ignore the related issue of detector inefficiency. 

For each triple define the average
\be\label{eq:epsilonbound1}
\epsilon(c, e^\alpha_i,e^\beta_j) :=  \frac13 \left(R[f_1|e^\alpha_i] + R[f_2|e^\beta_j] + R[f_3|c] \right) \, .
\ee
For each pair of states, chosen from the same basis, $e^\alpha_i$ and $e^\alpha_j$ ($i\ne j$), consider a measurement onto that basis, and define the average
\ba\label{eq:epsilonbound2}
\epsilon(e^\alpha_i,e^\beta_j) :=  \frac12 \left(R[e^\alpha_j | e^\alpha_i] + R[e^\alpha_i | e^\beta_j] \right)  \, .
\ea

Now consider an ontological model that predicts probabilities that coincide with the observed data. This means that for each preparation $\psi$ and outcome $g$, the probability predicted by the model satisfies
\be\label{eq:cond_prob_def}
P(g|\psi) := \int_\Lambda \xi_M(g|\lambda) \mu_{\psi} (\lambda) \mathrm{d}\lambda = R[g|\psi]\, .
\ee
For simplicity, the following assumes that the dimension $d$ is power prime. It is shown in Appendix 1 that in this case,
\ba
\label{eq:initial_inequality_main}
k d^2 \left( 1 - \sqrt{1 - \frac1d}\right) \le && \nonumber \\ 
1+   3 \sum_{\substack{\alpha<\beta \\ i,j}} && \epsilon(c, e^\alpha_i,e^\beta_j) +  2 \sum_{\substack{\alpha \\ i < j }} \epsilon(e^\alpha_i,e^\alpha_j) \,.\;\;\;
\ea

If we average the noise terms over all possible choices of measurement used in the experiment, defining
\be
\epsilon_1 := \frac{\sum_{\substack{\alpha<\beta,i,j}} \epsilon(c, e^\alpha_i,e^\beta_j)}{d^3(d-1)/2}, \qquad \epsilon_2 := \frac{\sum_{\substack{\alpha,i<j}} \epsilon(e^\alpha_i,e^\alpha_j)}{d^2(d-1)/2},
\ee
then
\begin{multline}
k d^2\left(1-\sqrt{1-1/d}\right)  \le  1 + \frac32 d^3(d-1) \epsilon_1 \\
+ d^2(d-1)\epsilon_2\,.
\end{multline}
Hence 
\ba
k & \le & \frac{1}{d}\left(1+ d^2(d-1)\left(\frac32 d \epsilon_1
+\epsilon_2 \right)\right)\left(1+\sqrt{1-1/d}\right) \nonumber \\
 & < & \frac{2}{d}+d^2 \left(3d \epsilon_1+2\epsilon_2 \right)\,.
\ea

For any value of $d\ge 4$ there exist small but non-zero values of $\epsilon_1$ and $\epsilon_2$ for which the experimentally determined bound $k < 1$ can be achieved. The result is therefore robust against small amounts of experimental noise and does not admit a finite precision loophole. In particular, a value of $k < 1$ is possible if the noise is bounded by: 
\be
3d \epsilon_1 +2\epsilon_2 < \frac{2}{d-1}\left(1-\sqrt{1-1/d}-\frac{1}{d^2}\right) \, .
\ee
Assuming $\epsilon=\epsilon_1=\epsilon_2$, this requires an error of $\epsilon < 0.0034$ for $d=4$ and even lower for higher dimensions. A high-precision measurement is required to achieve this, but it is one that is within the reach of the current state of the art using, for example, ion trap \cite{Brown2011, Gaebler2012} or magnetic resonance \cite{Ryan2009} technology.

\textit{Ruling out maximally epistemic models for $d=3$.}
The proof of Theorem~\ref{maintheorem} does not apply to the $d=3$ case, since mutually unbiased bases supply PP-incompatible triples only if $d\ge4$. It is, nonetheless, possible to rule out maximally epistemic models. The analysis of the noisy case turns out to be useful, because in $d=3$ one can construct a proof that makes use of triples of quantum states that are close to, rather than exactly, PP-incompatible. One can then apply an inequality analogous to Eq.~(\ref{eq:initial_inequality_main}). The details of this argument are given in Appendix~2. We obtain $k \leq 0.95$.

\textit{Conclusion.} We have considered ontological models for quantum systems, wherein a quantum state corresponds to a probability distribution over some set of ontic states. Such a model might be viewed as a schematic account of an underlying theory, more fundamental than quantum theory, but might equally be thought of as a classical simulation of quantum theory. Either way, it is interesting to investigate the constraints on such models, given that they reproduce quantum predictions. From an analysis of preparations and measurements on a single system, we have derived an upper bound on the extent to which probability distributions corresponding to distinct quantum states can overlap. An experimental challenge is to perform an experiment with sufficient precision that maximally $\psi$-epistemic models are ruled out. Finally, in prior work, Montina has established interesting connections between ontological models and communication complexity problems \cite{Montina2012}. It would be interesting to determine the relationship between our results and communication complexity.  

\acknowledgements{
\textit{Acknowledgements.} We would like to thank Andrew Briggs, Matthew Leifer, Stephanie Simmons and Christopher Timpson for helpful discussions during the development of this work. EGC received support from an Australian Research Council grant DE120100559. OJEM, RL are supported by the John Templeton Foundation. This work is supported by the CHIST-ERA DIQIP project, and an FQXi Large Grant ``Time and the Structure of Quantum Theory''.
}
\bibliography{epistemic_refs}

\newpage

\appendix*
\section{APPENDIX 1}\label{sec:appendix1}
We derive the inequality in Eq.~\eqref{eq:initial_inequality_main} as follows.

\subsection{The measure space}
For an ontic state space $\Lambda$, let us define the set ${\Gamma:=\Lambda \times \mathbb{R}^+}$. A positive function $f(\lambda) \ge 0$ that is integrable on $\Lambda$ 
(i.e. $\int_\Lambda f(\lambda)\,\mathrm{d}\lambda  < \infty $) defines a region $F$ of $\Gamma$ (i.e.~the area under $f$), where:
 \begin{equation}\label{eq:function_region}
 F:=\{(\lambda,x) \,|\, 0 \le x \le f(\lambda)\}
 \end{equation}
The union and intersection of two such regions are
\begin{eqnarray*}
F \cup G&=&\{(\lambda, x) \,|\, 0 \le x \le \max(f(\lambda),g(\lambda))\} \\
F \cap G&=&\{(\lambda, x) \,|\, 0 \le x \le \min(f(\lambda),g(\lambda))\}
\end{eqnarray*}

The volume measure on the space is $\mathrm{d}\gamma=\mathrm{d}\lambda \times \mathrm{d}x$, where $\textrm{d}x$ is the Lebesgue measure.  This gives
\[
\nu(F)=\int_F \,\mathrm{d}\gamma=\int_\Lambda \,\mathrm{d}\lambda  \int_0^{f(\lambda)} \,\mathrm{d}x =\int_\Lambda f(\lambda)
\,\mathrm{d}\lambda 
\]
and
\ba\label{eq:measure}
\nu(F \cup G )&=& \int \max(f(\lambda),g(\lambda))\,\mathrm{d}\lambda  \nonumber\\
\nu(F \cap G) &=& \int \min(f(\lambda),g(\lambda))\,\mathrm{d}\lambda 
\ea

For quantum states, a region $\Phi$ is defined by the corresponding epistemic states $\mu_\phi$. Eq.~\eqref{eq:function_region} becomes:
\be\label{eq:epistemicregion}
\Phi=\{(\lambda, x) \,|\,0\le x\le\mu_\phi(\lambda)\}.
\ee

\subsection{Bonferroni Inequality}
Now, using the first Bonferroni inequality \cite{Rohatgi2011}, on any measure space with measure $\nu$ we have:
\be \label{eq:Bonferroni}
\nu\left( \bigcup_k A_k \right) \ge \sum_{k} \nu\left( A_k \right) \\
-\sum_{k<k'} \nu\left(A_k\cap A_{k'} \right)
\ee
Consider a family of states $\{\ket{e^\alpha_i}\}_{i,\alpha}$, such that $\alpha$ labels a basis, and $i$ a basis element. Using Eq.~\eqref{eq:epistemicregion} each such state defines a region $E^\alpha_i$. Let $\ket{c}$ be a fixed state, with region $C$. We shall use Eq.~\eqref{eq:Bonferroni} with $A_k=A_{\alpha,i}=C \cap E^\alpha_i$, i.e.~we replace the index $k$ with the pair of indices $(\alpha,i)$.
\begin{multline} \label{eq:Bonferroni2}
\nu\left( \bigcup_{\alpha,i} C \cap E^\alpha_i  \right) \ge \sum_{\alpha,i} \nu\left( C \cap E^\alpha_i \right) \\
- \sum_{\substack{\alpha<\beta\\i,j}} \nu\left(C \cap E^\alpha_i \cap E^\beta_j \right)  - \sum_{\substack{\alpha=\beta\\i<j}} \nu\left(C \cap E^\alpha_i \cap E^\beta_j \right) 
\end{multline}
Note that from the normalization of $\mu_c$ we have
\be
\nu\left(\bigcup_{\alpha,i} C \cap E^\alpha_i\right)\le\nu(C)=1
\ee
Eq.~\eqref{eq:Bonferroni2} then becomes:
\begin{multline}
\sum_{\alpha,i} \nu\left(C \cap E^\alpha_i\right) \le 1 + \sum_{\substack{\alpha<\beta\\ i,j}} \nu\left(C \cap E^\alpha_i \cap E^\beta_j  \right) \\
+\sum_{\substack{\alpha=\beta\\i<j}} \nu\left(C \cap E^\alpha_i \cap E^\beta_j  \right)
\end{multline}
It will be useful to substitute:
\be
\nu\left(C \cap E^\alpha_i \cap E^\beta_j  \right) \le \nu\left(E^\alpha_i \cap E^\beta_j  \right)
\ee
for the cases where $\alpha=\beta$. Using Eq.~\eqref{eq:measure}, we then have:
\begin{multline}\label{eq:initial_inequality}
\sum_{\substack{\alpha,i}}  \int_\Lambda \min \left(\mu_c(\lambda),\mu_{e^\alpha_i}(\lambda) \right)\mathrm{d}\lambda \le   \\
1 +  \sum_{\substack{\alpha<\beta\\i,j}}\int_\Lambda \min\left(\mu_c(\lambda),\mu_{e^\alpha_i}(\lambda), \mu_{e^\beta_j}(\lambda)\right)\mathrm{d}\lambda    \\
  +  \sum_{\substack{\alpha\\i<j}}\int_\Lambda \min\left(\mu_{e^\alpha_i}(\lambda),\mu_{e^\alpha_j}(\lambda)\right)\mathrm{d}\lambda
\end{multline}

\subsection{Noise}
Consider a set of $n$ quantum states $\{\psi_1, \dots,\psi_n\}$ and a measurement $M$ with outcomes $\{f_1, \dots, f_n\}$. In an ontological model, the experimentally observed frequencies $R[f_i | \psi_j]$ are obtained by using the corresponding response functions $\xi_M(f_i|\lambda)$ and epistemic states $\mu_{\psi_j}$:
\be
\int_\Lambda \xi_M(f_i|\lambda)\mu_{\psi_j}(\lambda) \,\mathrm{d}\lambda = R[f_i | \psi_j].
\ee
Now, since $\forall i, \, \min_j \left(\mu_{\psi_j}(\lambda)\right) \le  \mu_{\psi_i}(\lambda)$,
\be
\int_\Lambda \xi_M(f_i|\lambda)\min_j \left(\mu_{\psi_j}(\lambda)\right) \,\mathrm{d}\lambda \le  R[f_i | \psi_i] \,.
\ee
From the normalization constraint ${\sum_{i=1}^n \xi_M(f_i|\lambda)=1}$, we then have:
\be\label{eq:bound} 
\int_\Lambda \min_j \left(\mu_{\psi_j}(\lambda)\right) \,\mathrm{d}\lambda \le \sum_{i=1}^n R[f_i | \psi_i].
\ee

Using this, along with Eqs.~\eqref{eq:epsilonbound1}, \eqref{eq:epsilonbound2} and \eqref{eq:cond_prob_def} in Eq.~\eqref{eq:initial_inequality}, we obtain
\begin{multline}
\label{eq:initial_inequality2}
\sum_{\substack{\alpha,i}}  \int_\Lambda \min \left(\mu_c(\lambda),\mu_{e^\alpha_i}(\lambda) \right)\mathrm{d}\lambda \le \\ 
1+   3 \sum_{\substack{\alpha<\beta \\ i,j}} \epsilon(c, e^\alpha_i,e^\beta_j) +  2 \sum_{\substack{\alpha \\ i < j }} \epsilon(e^\alpha_i,e^\alpha_j) \,.
\end{multline}
Then suppose, as in Theorem \ref{maintheorem}, that the ontological model satisfies $\omega_C(\mu_\psi,\mu_\phi) \geq k \,\omega_Q(\psi,\phi)$ for all pairs of states of a system of dimension $d\geq 4$, for some constant $k$. This implies that
\begin{equation}\label{overlapboundbykd}
\int_\Lambda \min \left(\mu_c(\lambda),\mu_{e^\alpha_i}(\lambda) \right)\mathrm{d}\lambda \ge k \left( 1 - \sqrt{1 - \frac1d}\right).
\end{equation}
If $d$ is power prime, then $d+1$ mutually unbiased bases can be found, as in the proof of Theorem~\ref{maintheorem}. In this case, substituting Eq.~(\ref{overlapboundbykd}) in Eq.~\eqref{eq:initial_inequality2} yields Eq.~\eqref{eq:initial_inequality_main} as required.
\section{APPENDIX 2}\label{sec:appendix2}

In the case $d=3$, there exist four mutually unbiased bases. However, if three vectors are chosen, each from a different basis in the mutually unbiased set, then Eq.~(\ref{eq:PP_incompatible}) of the main text shows that the resulting triple is not PP-incompatible. This means that the argument employed in the proof of Theorem~\ref{maintheorem} does not go through if $d=3$. It is still possible, however, to rule out maximally $\psi$-epistemic models, by adapting the results derived in Appendix~1 for the noisy case to the case in which triples of vectors are used that are approximately but not quite PP-incompatible. 

Consider the following three mutually unbiased bases in $d=3$:
\[
\ket{e^1_1}=
\begin{pmatrix}
1 \\
0 \\
0
\end{pmatrix}
,
\ket{e^1_2}=
\begin{pmatrix}
0 \\
1\\
0
\end{pmatrix}
,
\ket{e^1_3}=
\begin{pmatrix}
0 \\
0\\
1
\end{pmatrix},
\]
\[
\ket{e^2_1}=
\frac{1}{\sqrt{3}}
\begin{pmatrix}
1 \\
1 \\
\omega^2
\end{pmatrix}
,
\ket{e^2_2}=
\frac{1}{\sqrt{3}}
\begin{pmatrix}
1 \\
\omega^2\\
1
\end{pmatrix}
,
\ket{e^2_3}=
\frac{1}{\sqrt{3}}
\begin{pmatrix}
1 \\
\omega\\
\omega
\end{pmatrix},
\]
and
\[
\ket{e^3_1}=
\frac{1}{\sqrt{3}}
\begin{pmatrix}
1 \\
\omega\\
\omega^2
\end{pmatrix}
,
\ket{e^3_2}=
\frac{1}{\sqrt{3}}
\begin{pmatrix}
1 \\
1\\
1
\end{pmatrix}
,
\ket{e^3_3}=
\frac{1}{\sqrt{3}}
\begin{pmatrix}
1 \\
\omega^2\\
\omega
\end{pmatrix},
\]
where $\omega=e^{2\pi i/3}$.
Define the fixed normalised state:
\[
\ket{c}=
\begin{pmatrix}
 -0.374-0.236 i \\
 0.778-0.071 i \\
 0.018-0.441 i \\
\end{pmatrix}.
\]

For a triple $\ket{c},\ket{e^\alpha_i},\ket{e^\beta_j}$, and a measurement basis $\{f_1,f_2,f_3\}$, let
\[
\epsilon(c,e^\alpha_i,e^\beta_j) = \frac13 \left( P(f_1|e^\alpha_i) + P(f_2|e^\beta_j) + P(f_3|c) \right).
\]
This is similar to Eq.~(\ref{eq:epsilonbound1}) in the main text, except that the quantities on the right hand side are the probabilities predicted by quantum theory for a $d=3$ system, rather than observed frequencies in a noisy experiment. Suppose that an ontological model recovers the quantum predictions for the $d=3$ system. Then the derivation of the following inequality goes through exactly as that of Eq.~(\ref{eq:initial_inequality2}) in Appendix~1, except that experimental frequencies $R[\cdot | \cdot]$ are replaced with quantum probabilities $P( \cdot | \cdot )$, and the quantity $\epsilon(e^\alpha_i,e^\alpha_j)$ corresponds to $P(e^\alpha_i | e^\alpha_j)$, which is $0$ for $i\ne j$:
\begin{multline}
\label{eq:initial_inequality3}
\sum_{\substack{\alpha,i}}  \int_\Lambda \min \left(\mu_c(\lambda),\mu_{e^\alpha_i}(\lambda) \right)\mathrm{d}\lambda \le \\ 
1+   3 \sum_{\substack{\alpha<\beta \\ i,j}} \epsilon(c, e^\alpha_i,e^\beta_j) \,.
\end{multline}

If the triple $\ket{c},\ket{e^\alpha_i},\ket{e^\beta_j}$ is not PP-incompatible, then the quantity $\epsilon(c,e^\alpha_i,e^\beta_j)$ is non-zero for any choice of basis $\{f_1,f_2,f_3\}$. For each triple $\ket{c},\ket{e^\alpha_i},\ket{e^\beta_j}$, we found numerically the basis $\{f_1,f_2,f_3\}$ such that $\epsilon(c,e^\alpha_i,e^\beta_j)$ is minimized. These are shown in Tables \ref{tbl:first}, \ref{tbl:second} and \ref{tbl:third}, in which we display the optimal bases for $\alpha=1,\beta=2$, $\alpha=1,\beta=3$, and $\alpha=2,\beta=3$ respectively.

Reading from the last column of Table \ref{tbl:first} we have:
\[
3 \sum_{i,j}\epsilon(c,e^1_i,e^2_j)=0.2257\,.
\]
By including the data from Tables \ref{tbl:second} and \ref{tbl:third} we then have
\begin{multline}
3 \left[ \sum_{i,j}\epsilon(c,e^1_i,e^2_j)+ \sum_{i,j}\epsilon(c,e^1_i,e^3_j)+ \sum_{i,j}\epsilon(c,e^2_i,e^3_j) \right]= \\
0.649\,.
\end{multline}
For the states used we obtain $ \sum_{\alpha,i}\left(1-\sqrt{1-\born{e^\alpha_i}{c}}\right)=1.739$. Assuming that an ontological model satisfies $\omega_C(\mu_\psi,\mu_\phi) \geq k \,\omega_Q(\psi,\phi)$ for all pairs of states, and reproduces perfectly the quantum predictions, Eq.~\eqref{eq:initial_inequality3} gives $k\leq0.95<1$.

\newpage
\begin{widetext}

\begin{table}
\centering
\caption{Minimizing bases $\{f_i\}_i$, for the $\alpha=1,\beta=2$ family of states}\label{tbl:first}
\footnotesize
\[
\begin{array}{c||c|c|c||c}
(i,j) & f_1 & f_2 & f_3 & \epsilon(c,e^\alpha_i,e^\beta_j) \\
\hline
&  0 & -0.3439-0.6178 i & -0.6621-0.2482 i \\
 (1,1)  & 0.8171-0.5765 i & 0 & 0 & 0 \\
 & 0 & 0.3631-0.6068 i & 0.1161+0.6975 i  \\
 \hline
  & 0 & 0.5317-0.5285 i & 0.5013-0.4320 i \\
(1,2) & -0.3945-0.6556 i & 0.2415+0.3511 i & -0.2446-0.4161 i & 0 \\
 & -0.1740+0.6199 i & -0.1069+0.4949 i & 0.1607-0.5507 i  \\
 \hline
&  0.05556+0.08103 i & 0.2653+0.5440 i & 0.0222+0.7897 i \\
 (1,3) & -0.8215-0.1098 i & -0.0186-0.4206 i & -0.0416+0.3662 i & 0.02280 \\
 & -0.4846+0.2620 i & 0.4961+0.4586 i & -0.1328-0.4716 i \\
\hline
 &0.2686+0.1752 i & -0.5579+0.0032 i & -0.7647-0.0321 i \\
(2,1) &  -0.01766+0.06827 i & 0.7702-0.2381 i & -0.5608+0.1751 i & 0.02046 \\
& 0.9430+0.0539 i & 0.1924-0.0423 i & 0.2096-0.1583 i \\
\hline
& 0.8190+0.5614 i & 0.06326-0.05514 i & 0.05078-0.06618 i \\
(2,2) & -0.05017-0.05876 i & 0.7839+0.1828 i & -0.0791-0.5829 i & 0.02854 \\
 & -0.03910-0.08064 i & 0.2974-0.5065 i & 0.7911+0.1453 i \\
\hline
&  -0.2272-0.8467 i & 0.2412+0.1173 i & 0.1067-0.3848 i \\
 (2,3) & 0.2018+0.0689 i & 0.7575+0.3690 i & -0.4565+0.1902 i & 0.1119\\
& 0.4148-0.1180 i & -0.1806-0.4306 i & -0.6536-0.4109 i \\
\hline
& -0.1365-0.2454 i & 0.3174-0.7208 i & -0.5095-0.2032 i \\
 (3,1) & -0.6944-0.6626 i & -0.0261+0.2290 i & 0.1600+0.0126 i & 0 \\
& 0 & 0.5715+0.0063 i & 0.0198+0.8203 i \\
\hline
& 0.7135+0.6071 i & 0.1176+0.2737 i & -0.0914-0.1587 i \\
(3,2) &  -0.0483-0.3463 i & -0.2521+0.7572 i & 0.0999-0.4804 i & 0\\
& 0 & 0.4121+0.3232 i & 0.7268+0.4443 i \\
\hline
&  0.2274-0.7062 i & 0.5788+0.0661 i & 0.3110+0.1162 i \\
(3,3) &  0.0076+0.6570 i & 0.5561-0.0538 i & 0.4974-0.0934 i & 0.04198\\
& 0.11815+0.06333 i & -0.1493-0.5711 i & 0.1659+0.7785 i 
\end{array}
\]
\end{table}

\begin{table}
\caption{Optimal bases for $\alpha=1,\beta=3$}\label{tbl:second}
\footnotesize
\[
\begin{array}{c||c|c|c||c}
(i,j) & f_1 & f_2 & f_3 & \epsilon(c,e^\alpha_i,e^\beta_j) \\
\hline
 & 0 & 0.6374-0.3061 i & -0.7064-0.0324 i \\
(1,1) &  -0.7439+0.6683 i & 0 & 0 & 0 \\
& 0 & 0.5838+0.3990 i & 0.3251+0.6279 i \\
 \hline
 &-0.005857-0.001555 i & 0.3537-0.5332 i & 0.0495-0.7668 i \\
(1,2) &  0.1079-0.4250 i & -0.0440+0.6884 i & -0.2572-0.5155 i & 0.0001107 \\
 &-0.6994-0.5644 i & -0.2999-0.1573 i & 0.2785-0.0045 i \\
 \hline
 &0.08300+0.06424 i & 0.5014-0.3472 i & 0.1033-0.7787 i \\
 (1,3) & -0.5131-0.2961 i & 0.0173+0.6084 i & -0.3589-0.3869 i & 0.02699\\
 &-0.5763+0.5531 i & -0.4724-0.1854 i & 0.2506-0.2041 i \\
\hline
 &0.5057-0.8375 i & 0.10296+0.07596 i & -0.1545-0.0510 i \\
(2,1) & 0.02134+0.06721 i & 0.1922-0.7829 i & -0.4138-0.4170 i & 0.02046 \\
 &0.0046-0.1945 i & -0.4809-0.3202 i & 0.7088-0.3550 i \\
\hline
& 0.01974-0.01750 i & -0.1549-0.5597 i & -0.4182-0.6980 i \\
(2,2) &  0.09781+0.04938 i & 0.4684+0.6598 i & -0.4498-0.3619 i & 0.04659 \\
 &0.9804-0.1617 i & -0.08905-0.01351 i & 0.06312+0.02456 i \\
\hline
 &0.3240-0.1034 i & -0.3572+0.3563 i & -0.7926-0.0401 i \\
(2,3) &  -0.1937-0.0011 i & -0.5076-0.6591 i & -0.1714+0.4909 i & 0.09913 \\
 &0.2916+0.8728 i & 0.2304-0.0156 i & -0.1203+0.2923 i \\
\hline
 &0.5078+0.6160 i & 0.1906+0.4559 i & -0.3061-0.1575 i \\
(3,1) &  -0.2341+0.5549 i & 0.4183-0.5042 i & 0.0031+0.4563 i & 0 \\
& 0 & -0.5270-0.2212 i & -0.7834+0.2440 i \\
\hline
 &0.0329-0.6609 i & -0.2892+0.4989 i & -0.1802+0.4439 i \\
(3,2) & 0.7188-0.2129 i & 0.4983+0.0983 i & 0.4239+0.0159 i & 0.00006005 \\
 &0.002024-0.003977 i & -0.2098-0.6045 i & 0.3559+0.6811 i \\
\hline
 &0.1364+0.5608 i & 0.5771+0.3010 i & -0.4290-0.2435 i \\
(3,3) & 0.0167-0.8130 i & 0.3595+0.2569 i & -0.2174-0.3103 i & 0.01415 \\
 &-0.07506+0.01027 i & 0.6086+0.1037 i & 0.7462+0.2374 i \\
\end{array}
\]
\end{table}

\begin{table}
\caption{Optimal bases for $\alpha=2,\beta=3$}\label{tbl:third}
\footnotesize
\[
\begin{array}{c||c|c|c||c}
(i,j) & f_1 & f_2 & f_3 & \epsilon(c,e^\alpha_i,e^\beta_j) \\
\hline
& 0.5526-0.5319 i & 0.3382+0.1860 i & -0.4930-0.1402 i \\
 (1,1) & -0.5183+0.2811 i & 0.6172+0.3289 i & -0.2454-0.3212 i &0 \\
& 0.2344-0.0957 i & 0.6014+0.0158 i & 0.7443+0.1414 i \\
 \hline
 &-0.2839+0.4207 i & -0.11852-0.00351 i & -0.1025+0.8472 i \\
 (1,2) & -0.3820-0.3959 i & 0.7436+0.1541 i & 0.2789+0.2070 i &0.0001284 \\
& -0.3617-0.5557 i & -0.6194-0.1599 i & 0.3721+0.1124 i \\
 \hline
& 0.6414+0.2288 i & 0.2473+0.4718 i & -0.4496-0.2245 i \\
 (1,3) & -0.6346-0.3339 i & 0.2927+0.5058 i & -0.2340-0.2992 i & 0.02836 \\
& 0.0596+0.1361 i & 0.3784+0.4812 i & 0.6818+0.3720 i \\
\hline
& 0.6128-0.4277 i & -0.4496-0.2071 i & 0.4062+0.1773 i \\
 (2,1) & 0.6001+0.2710 i & 0.1370+0.5915 i & 0.0270-0.4439 i & 0\\
&  -0.07809+0.04348 i & 0.5353-0.3158 i & 0.7219-0.2907 i \\
\hline
& 0.6699+0.1325 i & -0.2021-0.5916 i & 0.3611+0.1114 i \\
 (2,2) & 0.2754+0.6738 i & -0.3359+0.4221 i & -0.1302-0.4027 i & 0\\
 & 0.05128-0.03404 i & 0.5380+0.1695 i & 0.6515-0.5036 i \\
\hline
& 0.2226-0.0614 i & 0.5312-0.4491 i & 0.1305-0.6677 i \\
 (2,3) & -0.6241+0.1562 i & 0.1084+0.4556 i & -0.3309-0.5073 i & 0.01016 \\
& -0.2877+0.6708 i & -0.3835-0.3871 i & 0.3932-0.1256 i \\
\hline
& 0.5453-0.1345 i & 0.2223-0.1133 i & -0.7708-0.1677 i \\
 (3,1) & -0.4957-0.4368 i & 0.3686+0.5809 i & -0.2060-0.2186 i & 0.04370\\
 & 0.4810-0.1291 i & 0.6644+0.1519 i & 0.4802+0.2385 i \\
\hline
& 0.3333+0.4625 i & 0.1472-0.7084 i & -0.0102-0.3891 i \\
 (3,2) & 0.1160+0.4700 i & -0.1432+0.6284 i & -0.2557-0.5337 i & 0.02959\\
&  0.5308-0.3986 i & 0.1481+0.1977 i & 0.6191-0.3393 i \\
\hline
& 0.3768-0.2046 i & -0.0433-0.5177 i & 0.2192+0.7058 i \\
 (3,3) & -0.6513-0.2318 i & 0.5348-0.4735 i & 0.08950-0.06175 i & 0.1035 \\
& 0.5662+0.1325 i & 0.3389-0.3242 i & 0.3229-0.5811 i \\
\end{array}
\]
\end{table}

\end{widetext}

\end{document}